\documentclass[letterpaper, 10 pt, conference]{ieeeconf}
\IEEEoverridecommandlockouts
\usepackage{amsmath}
\usepackage{mathtools}
\usepackage{xfrac}
\usepackage{amssymb}
\usepackage{xfrac}
\usepackage{dsfont}
\usepackage{graphicx}
\usepackage[dvipsnames]{xcolor}
\usepackage{booktabs}
\usepackage[font=small,labelfont=bf]{caption}
\usepackage{macros}
\usepackage{cite}
\usepackage{url}
\usepackage{tikz}
\usetikzlibrary{arrows,arrows.meta,backgrounds,decorations,decorations.pathmorphing,positioning,fit,automata,shapes,patterns,plotmarks,calc}
\usepackage[ruled,vlined,linesnumbered]{algorithm2e}
\usepackage{newtxtext}
\usepackage{newtxmath}
\usepackage{bm}
\usepackage[small]{subfigure}
\makeatletter
\let\NAT@parse\undefined
\makeatother
\usepackage[colorlinks=false, hidelinks, backref=false]{hyperref}
\usepackage{cleveref}
\graphicspath{{figs/}}
\newif\ifarxiv
\arxivtrue
\ifdefined\buildarxiv\arxivtrue\fi
\title{\LARGE \bf
Differentiable Ternary Temporal Logic Semantics \\ with Polynomial Surrogate Networks
}
\author{Ryan Matheu, Sai Sandeep Damera, John S. Baras and Calin Belta%
\thanks{The authors are with the Institute for Systems Research, University of Maryland, College Park, USA. Emails: \{\texttt{rmatheu}, \texttt{sdamera}, \texttt{baras}, \texttt{calin}\}\texttt{@umd.edu}.}%
}
\definecolor{goldencheese}{HTML}{F8D3A0}
\definecolor{exblue}{HTML}{2E5F7F}
\definecolor{exred}{HTML}{DA364B}
\definecolor{exgreen}{HTML}{248C2F}
\begin{document}
\maketitle
\thispagestyle{empty}
\pagestyle{empty}
\begin{abstract}
    Temporal logic provides designers a formal reasoning tool for specifying complex spatio-temporal behaviors and tasks for autonomous systems. Many techniques exist for control synthesis according to temporal logic specifications, spanning a large variety of systems, including robotics. However, the vast majority of implementations are limited to Boolean temporal logics. Boolean temporal logic does not have an innate mechanism for quantifying abstention, rather verdicts are either definitively \textit{True} or \textit{False}. Recent work demonstrates that ternary logic is a capable formalism for reasoning about robotic behaviors, with the natural ability to classify uncertainty with the \textit{Unknown} literal. Ternary temporal logic is still in its infancy, with work limited to synthesis for linear systems and monitoring. This work proposes a probabilistic relaxation of ternary temporal logic that allows for value and gradient computation through the same network while retaining exact verdicts on actual data. The relaxed network can be used as an objective in gradient-based control synthesis for both offline and online robotic control applications. We demonstrate the utility of our framework with two robotic manipulator tasks, one offline and one online, involving intermediate checking and sequencing of subtasks to demonstrate the benefit of our approach in terms of correct and timely execution.
\end{abstract}
\section{Introduction\label{sec:introduction}}
    Temporal logics like Signal Temporal Logic (STL)~\cite{maler2004monitoring} have become prominent formalisms for specifying tasks for autonomous systems thanks to their rich expressivity, ability to specify time-dependent behaviors, and real-valued robustness semantics~\cite{donze2010robust}. Temporal Behavior Trees (TBTs)~\cite{schirmer2024temporal} extend STL with the \textit{Sequence} and \textit{Selector} operators from behavior trees~\cite{colledanchise2018behavior}, allowing designers to specify long-horizon tasks as a sequence of subtasks with fallbacks. However, the interpretation of STL and TBT formulas is binary, meaning the satisfaction of a formula is either \textit{True} or \textit{False}. This is limiting for robotic applications where a robot must reason about its task while it is still being executed. For example, in model predictive control (MPC), the satisfaction of a formula is evaluated over a partial trajectory ending at the prediction horizon. A task that has not been completed by the end of the horizon is neither satisfied nor violated, yet Boolean semantics force a verdict of \textit{False}.

    Ternary logic offers a natural remedy to this problem. Kleene's strong logic $K_3$~\cite{kleene1952introduction} introduces a third truth value \textit{Unknown} to represent the case where there is not enough information to determine whether a proposition is \textit{True} or \textit{False}. In previous work~\cite{matheu2026ternary}, we formulate the qualitative ternary semantics of STL and TBTs over partial trajectories, and demonstrate their utility as constraints in optimal control problems solved via mixed-integer programming. Mixed-integer formulations are exact, but suffer from high complexity as the horizon and number of subtasks grow, and are often restricted to linear or piecewise-linear dynamics~\cite{raman2014model, sadraddini2015robust}. Gradient-based methods using smooth approximations of the real-valued STL quantitative semantics scale to nonlinear systems~\cite{pant2017smooth, gilpin2020smooth, mehdipour2019arithmetic, haghighi2019control}, but do not naturally quantify abstention and have not been implemented for the TBT operators. To the best of our knowledge, there is no smooth representation of the ternary semantics of temporal logic, so its suitability for gradient-based control synthesis remains open.

    \begin{figure}[t]
        \centering
\definecolor{ovGreenF}{HTML}{E4F1DA}\definecolor{ovGreenD}{HTML}{3E7A2E}%
\definecolor{ovBlueF}{HTML}{DCE8F4}\definecolor{ovBlueD}{HTML}{2E5F7F}%
\definecolor{ovOrgF}{HTML}{FCE4CA}\definecolor{ovOrgD}{HTML}{D2772A}%
\definecolor{ovPurF}{HTML}{F3E0F2}\definecolor{ovPurD}{HTML}{8E3E8C}%
\definecolor{ovGreyF}{HTML}{F1F1F1}\definecolor{ovGreyD}{HTML}{7A7A7A}%
\definecolor{ovGrad}{HTML}{B23A3A}%
\definecolor{ovF}{HTML}{DA364B}\definecolor{ovU}{HTML}{A3A3A3}\definecolor{ovT}{HTML}{248C2F}%
\begin{tikzpicture}[
    >={Stealth[length=4pt,width=3pt]},
    font=\footnotesize,
    inp/.style={draw=ovGreenD, fill=ovGreenF, rounded corners=2pt, line width=0.6pt,
        minimum width=1.5cm, minimum height=0.56cm, inner sep=0pt},
    cst/.style={draw=ovGreyD, fill=ovGreyF, rounded corners=2pt, line width=0.6pt, dashed,
        minimum width=1.5cm, minimum height=0.56cm, inner sep=0pt},
    andg/.style={draw=ovBlueD, fill=ovBlueF, rounded corners=1.5pt, line width=0.6pt,
        minimum width=0.55cm, minimum height=0.46cm, inner sep=0pt},
    org/.style={draw=ovOrgD, fill=ovOrgF, rounded corners=1.5pt, line width=0.6pt,
        minimum width=0.55cm, minimum height=0.46cm, inner sep=0pt},
    wire/.style={->, draw=black!65, line width=0.5pt},
    pics/trithist/.style n args={3}{code={
        \fill[ovF] (-0.2,0) rectangle (-0.075,{#1*0.36});
        \fill[ovU] (-0.0625,0) rectangle (0.0625,{#2*0.36});
        \fill[ovT] (0.075,0) rectangle (0.2,{#3*0.36});
        \draw[black!55, line width=0.3pt] (-0.24,0) -- (0.24,0);
    }},
]

\node[anchor=west, text=black!70, inner sep=0pt] at (0.05,0.55) {random trits $(m,q)$};
\foreach \name/\lab/\y/\pf/\pu/\pt in {
        a0/a_0/-0.68/0.018/0.292/0.690,
        a1/a_1/-1.36/0.168/0.664/0.168,
        b1/b_1/-2.04/0.012/0.219/0.769,
        b2/b_2/-2.72/0.018/0.292/0.690} {
    \node[inp] (\name) at (0.8,\y) {};
    \node at ($(\name.west)+(0.3,0)$) {$\lab$};
    \pic[scale=1.1] at ($(\name.center)+(0.33,-0.2)$) {trithist={\pf}{\pu}{\pt}};
}
\foreach \name/\y in {u1/0, u2/-3.4} {
    \node[cst] (\name) at (0.8,\y) {};
    \node[text=black!75] at ($(\name.west)+(0.3,0)$) {$U$};
    \pic[scale=1.1] at ($(\name.center)+(0.33,-0.2)$) {trithist={0}{1}{0}};
}

\node[org]  (g1) at (2.55,-0.34) {$\vee$};
\node[org]  (g2) at (2.55,-1.02) {$\vee$};
\node[andg] (g3) at (2.55,-2.38) {$\wedge$};
\node[andg] (g4) at (2.55,-3.06) {$\wedge$};
\node[andg] (h1) at (3.75,-1.0)   {$\wedge$};
\node[andg] (h2) at (3.75,-2.4)   {$\wedge$};
\node[org]  (o)  at (4.85,-1.7)  {$\vee$};

\foreach \a/\b in {u1/g1, a0/g1, a0/g2, a1/g2, b1/g3, b2/g3, b2/g4, u2/g4}
    { \draw[wire] (\a.east) -- (\b.west); }
\foreach \a/\b in {g1/h1, g3/h1, g2/h2, g4/h2, h1/o, h2/o}
    { \draw[wire] (\a.east) -- (\b.west); }

\node[draw=ovPurD, fill=ovPurF, rounded corners=2pt, line width=0.6pt,
      minimum width=2.3cm, minimum height=2.0cm] (out) at (6.85,-1.7) {};
\draw[wire] (o.east) -- (out.west |- o.east);
\node[anchor=north, inner sep=0pt] at ($(out.north)+(0,-0.14)$) {$Z^\varphi_{0,2}$};
\pic[scale=2.0] at ($(out.center)+(0,-0.42)$) {trithist={0.001}{0.634}{0.366}};
\foreach \l/\dx in {$F$/-0.325, $U$/0, $T$/0.325}
    { \node[text=black!70, font=\scriptsize] at ($(out.center)+(\dx,-0.6)$) {\l}; }
\node[text=ovT, font=\scriptsize] at ($(out.center)+(0.425,0.0)$) {$0.37$};
\node[anchor=north] at ($(out.south)+(0,-0.06)$) {$\Pr(Z^\varphi{=}{+}1)=\tfrac12(m+q)$};

\node[draw=ovGrad, fill=ovGrad!8, rounded corners=2pt, line width=0.6pt,
      minimum width=1.5cm, minimum height=0.5cm, inner sep=0pt] (u) at (0.8,-4.2) {controls $\bm u$};
\draw[->, ovGrad, line width=0.7pt, densely dashed, rounded corners=4pt]
    ($(out.south)+(0,-0.5)$) |- (u.east);
\node[text=ovGrad, fill=white, inner sep=1.5pt] at (3.9,-4.2)
    {$\nabla_{\bm u}\log\Pr(Z^\varphi{=}{+}1)$};
\end{tikzpicture}
        \caption{Ternary temporal logic-equivalent logic gate network for $\varphi=\seq(\lozenge_{[0,1]}\mu_1,\,\square_{[0,1]}\mu_2)$ at partial horizon $t_2=2$. Each predicate is relaxed into a random trit, and its first two moments $(m,\,q)$ are propagated through the AND and OR gates in closed form, giving the distribution of the verdict $Z^\varphi_{0,2}$. Search windows that extend beyond the partial horizon contribute constant \textit{Unknown} inputs. The likelihood of satisfaction provides a smooth objective for gradient-based control synthesis.}
        \label{fig:overview}
    \end{figure}
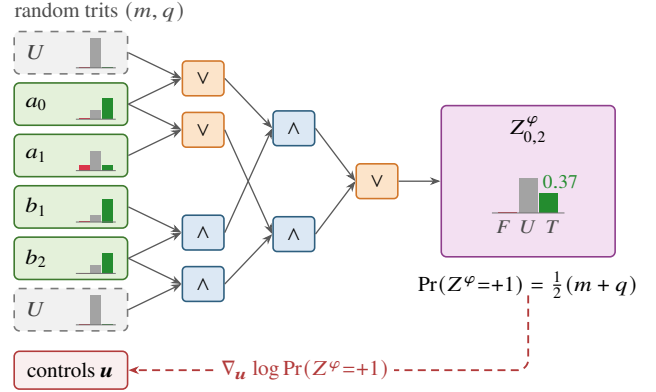

    This work builds on the notion of using ternary logic gate networks as surrogates for temporal logic. In~\cite{damera2026polynomial}, we showed that ternary logic functions admit unique polynomial surrogates that are exact on the trits $\{{-}1,\,0,\,{+}1\}$, while being smooth by construction. In this paper, we demonstrate that the satisfaction of any ternary temporal logic formula over a partial trajectory can be unrolled into a Differentiable Ternary Logic Gate Network (DTLGN), providing a useful objective in gradient-based control synthesis.

    Our primary contributions in this work are:
    \begin{enumerate}
        \item
            We show that the satisfaction of any ternary temporal logic formula can be represented exactly as a DTLGN.
        \item
            We introduce a moment propagation scheme through DTLGNs that yields a smooth likelihood of satisfaction.
        \item
            We demonstrate gradient-based control synthesis from ternary temporal logic specifications on a 7-DoF robotic manipulator for both open-loop and closed-loop (MPC) tasks.
    \end{enumerate}

    This paper is organized as follows: Section~\ref{sec:related-work} describes work related to smooth temporal logic semantics and ternary logic gate networks. In Section~\ref{sec:preliminaries}, we provide the necessary preliminaries on ternary logic, temporal logic, and polynomial surrogates for logical functions. We formulate the main problem in Section~\ref{sec:problem-formulation} and present our technical solution in Section~\ref{sec:encodings}. Section~\ref{sec:experiments} demonstrates our framework with two control tasks using a simulated robotic manipulator.
\section{Related Work\label{sec:related-work}}
    \subsection{Temporal Logic and Smooth Qualitative Semantics}
        Signal Temporal Logic (STL)~\cite{maler2004monitoring} is a logical system for specifying spatio-temporal properties of real-valued signals, reasoning about the satisfaction of a formula with respect to a system trajectory. In addition to the qualitative (Boolean) semantics, STL has well-defined quantitative (robustness) semantics that measure how strongly a signal satisfies or violates a formula~\cite{donze2010robust}. However, computation of the robustness of a formula consists of nested $\min$ and $\max$ functions, meaning the nominal robustness measure is not smooth. Temporal Behavior Trees (TBTs)~\cite{schirmer2024temporal} extend STL with the \textit{Sequence} and \textit{Selector} operators derived from behavior trees. These new temporal operators segment a trajectory into consecutive subtasks or fallbacks, and are equipped with a corresponding robustness measure. In previous work~\cite{matheu2026ternary}, we formulate the ternary logic qualitative semantics for STL and TBTs and demonstrate their utility as constraints in optimal control problems solved via mixed-integer programming.
        
        Closed-loop control synthesis with STL formulas as constraints on the system trajectory is traditionally posed as an online mixed-integer program. The Boolean satisfaction of a formula is encoded as a hard constraint, or the robustness is added to the cost as an objective to be maximized~\cite{raman2014model, sadraddini2015robust}. Such formulations are exact, but suffer from high complexity. To enable gradient-based optimization, the non-smooth robustness measure with is fitted with smooth approximations of $\min$ and $\max$~\cite{pant2017smooth}, or alternative smooth under-approximations to preserve soundness~\cite{gilpin2020smooth}. Other approaches to smooth robustness semantics include a reformulation using the arithmetic-geometric mean~\cite{mehdipour2019arithmetic} or a cumulative formulation~\cite{haghighi2019control} allowing for smooth semantics by construction. All of these approaches focus on a smooth \textit{quantitative} measure of satisfaction. In contrast, we smoothen the \textit{qualitative} ternary verdict itself. Our surrogate agrees exactly with the hard semantics on the trits, admits an \textit{Unknown} verdict over partial trajectories, and covers STL as well as the TBT operators.

    \subsection{Ternary Logic Gate Networks and Reparameterization}
        In previous work~\cite{damera2026polynomial}, we demonstrate that general $k$-ary ternary logic functions can be constructed as logic gate networks with learnable parameters. Each neuron represents a two-input ternary logic gate with nine learnable coefficients, covering the full gate library consisting of $3^{3^2}$ unique functions. As an extension to that work, we demonstrate that Differentiable Ternary Logic Gate Networks (DTLGNs) are capable of learning to causally predict STL verdicts online, choosing to abstain from producing a verdict in the absence of sufficient information~\cite{sandeep2026stability}. This work extends the notion of using ternary logic gate networks as surrogates for temporal logic monitoring, creating a representation for the satisfaction of a formula that is smooth by construction and provides exact verdicts on actual input traces.

        The Gumbel-Softmax~\cite{jang2016categorical} and Concrete~\cite{maddison2016concrete} distributions provide a continuous relaxation of categorical random variables, allowing discrete variables to be represented by a smooth, temperature-controlled distribution through which gradients can be propagated. In this work, our predicate relaxation follows from the same idea: logistic noise added to the predicate function turns each hard verdict into a random ternary variable, very much in the spirit of differentiable programming~\cite{blondel2024elements}, where differentiable relaxations such as smoothing programs have been interpreted in a probabilistic framework. Rather than sampling, we propagate the moments of these ternary variables analytically, which gives a deterministic objective and exact gradients.

\section{Preliminaries\label{sec:preliminaries}}

    \subsection{Ternary Temporal Logic}
        Kleene's strong logic ($K_3$)~\cite{kleene1952introduction} proposes a third truth value \textit{Unknown} to represent indeterminacy in computation or reasoning, concretizing the case when a proposition is neither \textit{True} nor \textit{False}.

        \begin{definition}[$K_3$ Syntax and Semantics]
            \label{def:k3}
            Let $P$ be a set of propositional variables that take values in the set $\{False,\,Unknown,\,True\}$ (abbr. $\{F,\,U,\,T\}$). A formula in $K_3$ is defined inductively as: 1) every $p\in P$ is a formula; 2) if $\varphi$ is a formula, then $\neg \varphi$ is a formula; 3) if $\varphi_1$ and $\varphi_2$ are formulas, then $\varphi_1\wedge \varphi_2$ and $\varphi_1\vee \varphi_2$ are formulas. Kleene's strong interpretation of the connectives $\neg$, $\wedge$, and $\vee$ is defined as:
            
            \begin{center}
                \begin{tabular}{c|c}
                    $\neg$ & \\
                    \hline
                    $F$ & $T$ \\
                    $U$ & $U$ \\
                    $T$ & $F$
                \end{tabular}
                \quad
                \begin{tabular}{c|ccc}
                    $\wedge$ & $F$ & $U$ & $T$ \\
                    \hline
                    $F$ & $F$ & $F$ & $F$ \\
                    $U$ & $F$ & $U$ & $U$ \\
                    $T$ & $F$ & $U$ & $T$
                \end{tabular}
                \quad
                \begin{tabular}{c|ccc}
                    $\vee$ & $F$ & $U$ & $T$ \\
                    \hline
                    $F$ & $F$ & $U$ & $T$ \\
                    $U$ & $U$ & $U$ & $T$ \\
                    $T$ & $T$ & $T$ & $T$
                \end{tabular}
            \end{center}
        \end{definition}
        Let $\mathbb{K}=\{F,\,U,\,T\}$ be the set of ternary truth constants. $K_3=(\mathbb{K},\,\wedge,\,\vee,\,\neg,\,F,\,T)$ forms a Kleene algebra\footnote{A De Morgan algebra whose elements satisfy $x\wedge\neg x\leq y\vee\neg y$. This is distinct from the Kleene algebra generalizing regular languages.}~\cite{kaarli2000polynomial} with $F<U<T$. Negation is an involution satisfying $\neg(x\wedge y)\,\leftrightarrow\,\neg x\vee\neg y$ for all $x,\,y\in\mathbb{K}$. Logical implication retains its usual definition: $x\rightarrow y:=\neg x\vee y$. Where $K_3$ diverges from Boolean algebras is in its failure to satisfy the complement laws, i.e., $x\wedge\neg x= F$ and $x\vee\neg x= T$. These laws fail in the case $x=U$.
        
        We consider the discrete-time domain as the set of all nonnegative integers: $\mathbb{T}:=\{t\in\mathbb{Z}: t\geq0\}$. The state of a dynamical system is represented by a real vector $x_t\in\mathbb{R}^n$ for time step $t\in\mathbb{T}$. A system trajectory is an infinite tuple of states $\bm{x}=(x_0,\,x_1,\,x_2,\dots)$, and a partial trajectory is a contiguous subsequence of an infinite system trajectory. Partial trajectories have a final time step, which we refer to as the partial trajectory horizon, or partial horizon. A partial trajectory ending at partial horizon $t_2$ is a finite tuple $\bm{x}_\text{partial}=(x_{0},\,x_{1},\dots,\,x_{t_2-1},\,x_{t_2})$ where $t_2\geq0$, i.e., a partial trajectory always contains at least one point. Because we are working with discrete time indices, a time interval is formally defined as the set $[a,\,b]:=\{t\in\mathbb{T}: a\leq t\leq b,\,a,b\in\mathbb{T}\}$. Note that this definition includes time intervals that contain a single time step, as is the case when $a=b$.
        
        The Boolean logic semantics of STL~\cite{maler2004monitoring} are defined over infinite trajectories, where the satisfaction of a formula is interpreted as either \textit{True} or \textit{False}. TBTs~\cite{schirmer2024temporal} extend STL by introducing new temporal operators inspired by the BT control nodes. Previously, we defined the qualitative ternary logic semantics of STL and TBTs and their equivalent mixed-integer encodings, quantifying satisfaction over partial system trajectories~\cite{matheu2026ternary}. Given a partial trajectory, it may be the case that there is not enough information to definitively determine the satisfaction of a formula. In this case we interpret the satisfaction of the formula to be \textit{Unknown}. A formula written in the Ternary Temporal Logic (TTL) fragment of STL and TBTs is defined inductively over the following grammar:
        \begin{align*}
            \varphi := \mu \mid {}& \neg\varphi \mid \varphi_1 \wedge \varphi_2 \mid \square_{[a,b]}\,\varphi \mid \lozenge_{[a,b]}\,\varphi \\
            \mid {}& \seq{(\varphi_1, \varphi_2)} \mid \sel{(\varphi_1, \varphi_2)}
        \end{align*}
        where $[a,\,b]$ is a time interval. Formulas are recursively defined starting with the signal predicates ($\mu$) and composed together using logical connectives (${\neg},\,{\wedge}$) and temporal operators. Other familiar logical connectives such as \textit{disjunction} (${\vee}$), \textit{implication} (${\rightarrow}$), and \textit{equivalence} (${\leftrightarrow}$) are defined using De Morgan's laws. $\square_{[a,\,b]}$ is the unary \textit{always} STL operator and $\lozenge_{[a,\,b]}$ is the unary \textit{eventually} STL operator. $\seq$ is the binary \textit{Sequence} TBT operator and $\sel$ is the binary \textit{Selector} TBT operator. Ternary signal predicates are functions $\mu:\mathbb{R}^n\to\{F,\,U,\,T\}$ with the general form:
        \begin{equation}
            \label{eq:pred}
            \mu(x_t):=\begin{cases}
                \;T, & f(x_t)\geq\delta, \\
                \;U, & -\delta<f(x_t)<\delta, \\
                \;F, & f(x_t)\leq-\delta,
            \end{cases}
        \end{equation}
        where $f:\mathbb{R}^n\to\mathbb{R}$ is a Lipschitz continuous function and $\delta>0$ is a user-defined uncertainty threshold. In the case where there is no uncertainty threshold, ternary signal predicates fall back to the Boolean sense where $\mu(x_t)=T\rightarrow f(x_t)\geq0$ and $\mu(x_t)=F\rightarrow f(x_t)<0$. As an example, a specification formulated as $\seq{(\lozenge_{[0,\,5]}(y\leq2),\,\lozenge_{[0,\,10]}\square_{[0,\,2]}(y\geq5))}$ states that $y$ must first become less than or equal to $2$ within $6$ time steps, then following completion of the first subformula, $y$ must become greater than or equal to $5$ within $11$ time steps and remain at or above $5$ for $3$ time steps.
        
        \begin{definition}[TTL Qualitative Semantics]
            \label{def:qual}
            Let the statement $\bm{x},\,t_1,\,t_2\models\varphi$ be the satisfaction of TTL formula $\varphi$ at time step $t_1$ over trajectory $\bm{x}$ with partial horizon $t_2$. The satisfaction at time step $t_1\leq t_2$ takes values in the ternary set of truth constants: $(\bm{x},\,t_1,\,t_2\models\varphi)\in\{F,\,U,\,T\}$. Given a trajectory $\bm{x}$ and a partial horizon $t_2$, the satisfaction of an STL formula at time step $t_1\leq t_2$ is defined inductively starting with the untimed connectives:
            \[
                \begin{aligned}
                    \bm{x},\,t_1,\,t_2\models\mu\;&\leftrightarrow\;\mu(x_{t_1}), \\
                    \bm{x},\,t_1,\,t_2\models\neg\varphi\;&\leftrightarrow\;\neg(\bm{x},\,t_1,\,t_2\models\varphi),\\
                    \bm{x},\,t_1,\,t_2\models\varphi_1\wedge\varphi_2\;&\leftrightarrow\;(\bm{x},\,t_1,\,t_2\models\varphi_1)\wedge(\bm{x},\,t_1,\,t_2\models\varphi_2),
                \end{aligned}
            \]
            where the logical connectives adhere to the rules in Def.~\ref{def:k3}. The notations $\neg(\bm{x},\,t_1,\,t_2\models\varphi)$ and $\bm{x},\,t_1,\,t_2\not\models\varphi$ are used interchangeably. Let the \textit{search window} of a timed operator be the time interval in which the operand is evaluated. For example, the statement $\bm{x},\,t_1,\,t_2\models\square_{[a,\,b]}\varphi$ is determined by evaluating $\varphi$ over the interval $[a+t_1,\,b+t_1]$, where $[a+t_1,\,b+t_1]$ is the search window. The timed operators elicit branching behavior dependent on the relationship between the search window and the partial horizon. The satisfaction of $\bm{x},\,t_1,\,t_2\models\square_{[a,\,b]}\varphi$ is equivalently:
            \[
                \begin{cases}
                    \;\displaystyle\bigwedge_{\tau=a+t_1}^{b+t_1}{(\bm{x},\,\tau,\,t_2\models\varphi)},\qquad\qquad&b+t_1\leq t_2, \\[5pt]
                    \;U\wedge\displaystyle\bigwedge_{\tau=a+t_1}^{t_2}{(\bm{x},\,\tau,\,t_2\models\varphi)},&b+t_1>t_2,                    
                \end{cases}
            \]
            and $\bm{x},\,t_1,\,t_2\models\lozenge_{[a,\,b]}\varphi$ is equivalently:
            \[
                \begin{cases}
                    \;\displaystyle\bigvee_{\tau=a+t_1}^{b+t_1}{(\bm{x},\,\tau,\,t_2\models\varphi)},\qquad\qquad&b+t_1\leq t_2, \\
                    \;U\vee\displaystyle\bigvee_{\tau=a+t_1}^{t_2}{(\bm{x},\,\tau,\,t_2\models\varphi)}, &b+t_1>t_2.
                \end{cases}
            \]
            The Sequence operator has the following semantics: $\bm{x},\,t_1,\,t_2\models\seq{(\varphi_1,\,\varphi_2)}$ is equivalently:
            \[
                \bigvee_{\tau=t_1}^{t_2-1}(\bm{x},\,t_1,\,\tau\models\varphi_1)\,\wedge\,(\bm{x},\,\tau+1,\,t_2\models\varphi_2),
            \]
            meaning there exists a partial trajectory satisfying $\varphi_1$ and a partial trajectory beginning at the next time step satisfying $\varphi_2$. The Selector operator has the following semantics: $\bm{x},\,t_1,\,t_2\models\sel{(\varphi_1,\,\varphi_2)}$ is equivalently:
            \[
                \bigvee_{\tau=t_1}^{t_2-1}(\bm{x},\,t_1,\,\tau\models\varphi_1)\vee(\bm{x},\,\tau+1,\,t_2\models\varphi_2).
            \]
            In other words, the Selector operator is satisfied if a partial trajectory satisfies $\varphi_1$ or the partial trajectory beginning at the next time step satisfies $\varphi_2$.

        \end{definition}

    \subsection{Polynomial Surrogates for Ternary Logic}
        Let the ternary truth constants be represented as \textit{trits}, integers taking values in the set $\{{-}1,\,0,\,{+}1\}$, where $True\mapsto {+}1$, $Unknown\mapsto 0$, and $False\mapsto {-}1$, thus preserving the order $F<U<T$. A real-valued surrogate function $\hat{f}:[{-}1,\,1]^2\to[{-}1,\,1]$ is said to represent a ternary logic function $f:\{{-}1,\,0,\,{+}1\}^2\to\{{-}1,\,0,\,{+}1\}$ if for all $z_1,\,z_2\in\{{-}1,\,0,\,{+}1\}$, the following holds true:
        \[
            \hat{f}(z_1,\,z_2)=f(z_1,\,z_2).
        \]
        Like Boolean functions that have multilinear representations, ternary functions admit unique real-valued multivariate polynomial representations whose individual terms have powers no greater than two~\cite{o2014analysis, damera2026polynomial}. By construction, these polynomial surrogates are $C^\infty$-smooth, and closed on $[{-}1,\,1]$. In fact, each of the $3^{3^2}$ two-input ternary logic functions $f:{\{{-}1,\,0,\,{+}1\}}^2\to\{{-}1,\,0,\,{+}1\}$ is represented by the polynomial $p_{\bm{w}} : \R^2 \to \R$ defined as:
        \[
            p_{\bm{w}}(x,\,y) = \bm{w}^\top\bm{m}(x,\,y),
        \]
        where $\bm{m}(x,\,y) = [1,\,x,\,y,\,xy,\,x^2,\,y^2,\,x^2y,\,xy^2,\,x^2y^2]^\top$ is a monomial basis and $\bm{w}\in \R^9$ are learnable coefficients. The $9\times9$ Vandermonde matrix $V$ evaluating these monomials at the grid points ${\{{-}1,\,0,\,{+}1\}}^2$ is invertible, so the mapping between coefficients $\bm{w}$ and truth table values $\bm{t}=V\bm{w}$ is a bijection~\cite{sandeep2026stability}. The real-valued polynomial surrogates corresponding to $K_3$ conjunction and disjunction are as follows:
        \begin{equation}
            \label{eq:polys}
            \begin{aligned}
                p_\wedge(x,\,y)&=\frac{1}{2}\left(
                    x+y+xy-x^2-y^2+x^2y^2
                \right), \\
                p_\vee(x,\,y)&=\frac{1}{2}\left(
                    x+y-xy+x^2+y^2-x^2y^2
                \right).
            \end{aligned}
        \end{equation}
        Given $z_1,\,z_2\in\{{-}1,\,0,\,{+}1\}$, the $K_3$ connectives are represented exactly by the polynomials above according to the mapping:
        \[
            \neg z\;:=\;-z,\quad z_1\wedge z_2:=p_\wedge(z_1,\,z_2),\quad z_1\vee z_2:=p_\vee(z_1,\,z_2),
        \]
        with general $k$-ary conjunction and disjunction defined as being right-nested:
        \[
            \bigwedge_{j=1}^{k}z_j:=p_\wedge\bigl(z_1,\,p_\wedge({\cdots},\,z_k)\bigr),\;\bigvee_{j=1}^{k}z_j:=p_\vee\bigl(z_1,\,p_\vee({\cdots},\,z_k)\bigr).
        \]

\section{Problem Formulation\label{sec:problem-formulation}}

    A discrete-time dynamical system has dynamics of the form:
    \begin{equation}
        \label{eq:sys-dyn}
        x_{t+1}=f(x_t,\,u_t)
    \end{equation}
    where $x_t\in\mathbb{R}^n$ is the state vector at time step $t\in\mathbb{T}$ and $u_t\in\mathcal{U}\subseteq\mathbb{R}^m$ is the control vector at time step $t$. We consider the optimal control of nonlinear systems of the form~\eqref{eq:sys-dyn}, where the task is provided as a TTL formula $\varphi$ on the system trajectory. The verdict trit $z^\varphi_{t_1,\,t_2}:=(\bm{x},\,t_1,\,t_2\models\varphi)\in\{{-}1,\,0,\,{+}1\}$ is piecewise constant with respect to the trajectory, meaning its gradient is zero almost everywhere and provides no information for gradient-based optimization. Mathematical optimization techniques like mixed-integer formulations avoid this issue by encoding the qualitative semantics as a set of mixed-integer constraints. However, direct optimization techniques suffer from high complexity and, for long horizon tasks, are limited to linear or piecewise-linear dynamics. We instead propose to lift the satisfaction of $\varphi$ to a smooth likelihood of satisfaction, retaining the exact ternary verdict and providing useful gradient information.

    \begin{prob}[Smooth TTL Semantics]
        \label{prob:smooth}
        Given a TTL formula $\varphi$ with satisfaction $z^\varphi_{t_1,\,t_2}$, determine a random verdict $Z^\varphi_{t_1,\,t_2}$ whose likelihood of satisfaction $\Pr(Z^\varphi_{t_1,\,t_2}={+}1)$ is smooth with respect to the predicate values $f(x_t)$, and recovers exact verdicts whenever the predicate trits are deterministic.
    \end{prob}

    A solution to Problem~\ref{prob:smooth} provides a sound objective for gradient-based control synthesis. As probabilities for cold-start gradient optimization may be very small and cause arithmetic underflow, we instead formulate the objective in log-space.

    \begin{prob}[TTL Control Synthesis]
        \label{prob:control}
        Given a system of the form~\eqref{eq:sys-dyn} with known initial state $\xi$, a horizon $N$, and a TTL formula $\varphi$, find a control sequence $\bm{u}=(u_0,\dots,u_{N-1})$ that solves:
        \begin{equation}
            \label{eq:ocp}
            \begin{aligned}
                \min_{\bm{u}}\;&\sum_{t=0}^{N-1}u_t^\top Ru_t-\log\Pr\bigl(Z^\varphi_{0,\,N}={+}1\bigr)\\
                \text{s.t.}\;&x_{t+1}=f(x_t,\,u_t),\quad u_t\in\mathcal{U},\quad x_0=\xi.
            \end{aligned}
        \end{equation}
    \end{prob}

    An optimal solution to~\eqref{eq:ocp} provides an open-loop control sequence with high likelihood of satisfying $\varphi$. Modern automatic differentiation software like JAX~\cite{jax2018github} provides the tooling required to compute analytical gradients of the objective with respect to the control input sequence $\bm{u}$.

\section{Polynomial Smoothing for TTL\label{sec:encodings}}    
    Exact agreement on predicate trits does not determine how a polynomial circuit behaves between those trits. For example, conjunction with an \textit{Unknown} input gives:
    \begin{equation}
        \label{psttl:eq:scalar-unknown}
        p_\wedge(x,\,0)=\tfrac12\bigl(x-x^2\bigr),\qquad
        \frac{\partial p_\wedge}{\partial x}(x,\,0)=\tfrac12-x.
    \end{equation}
    The scalar surrogate decreases as $x$ moves from $\sfrac12$ toward \textit{True}. Consequently, maximizing the scalar output can oppose increasing a predicate's truth value. \textit{Unknown} inputs occur structurally when temporal windows extend beyond the available trajectory, so this behavior is relevant to partial-horizon evaluation.
    
    We instead represent each relaxed predicate by a distribution over trits and propagate two moments through the circuit. The second moment distinguishes a certainly \textit{Unknown} trit from an equally weighted \textit{True}/\textit{False} mixture. Both have zero mean, but their second moments are zero and one, respectively. This distinction permits an objective that rewards \textit{True} without assigning partial credit to \textit{Unknown}.

    \subsection{TTL Networks}
        We now extend the representation of the $K_3$ connectives via polynomial surrogates to the full TTL language. A polynomial surrogate for a TTL formula is a function $f:{\{{-}1,\,0,\,{+}1\}}^k\to\{{-}1,\,0,\,{+}1\}$, where the $k$ inputs are derived from predicate valuations over time. That is, each TTL surrogate polynomial is a \textit{transducer}, mapping the valuation of input predicate trits to the satisfaction, violation, or abstention of a formula.
        \begin{example}
            Consider $\varphi=\seq\bigl(\lozenge_{[0,1]}\mu_1,\,\square_{[0,1]}\mu_2\bigr)$ with partial horizon $t_2=2$, and let $a_t=\mu_1(x_t)$ and $b_t=\mu_2(x_t)$. Unrolling the Sequence operator over $s\in[0,\,1]$ gives:
            \[
                z^\varphi_{0,2}=\bigl[(U\vee a_0)\wedge(b_1\wedge b_2)\bigr]\vee\bigl[(a_0\vee a_1)\wedge(U\wedge b_2)\bigr],
            \]
            where the constant $U$ literals arise from search windows that extend beyond the partial horizon. The corresponding DTLGN is shown in Fig.~\ref{fig:overview}.
        \end{example}
        Constructively, expanding the temporal operators over a fixed partial interval yields a finite $K_3$ expression over predicate trits and \textit{Unknown} constants. De Morgan's laws, double negation, and $\neg U=U$ allow this expression to be rewritten without changing its ternary value so that negation applies only to predicate inputs. This is its \emph{negation normal form} (NNF). The corresponding syntax tree has conjunctions and disjunctions at internal nodes and predicate literals or \textit{Unknown} constants at its leaves. We use this tree to construct the equivalent gate network inductively, from the leaves to the root.

        \begin{lemma}
            Given a TTL formula $\varphi$, the ternary satisfaction $z^\varphi_{t_1,\,t_2}$ can be represented exactly as a DTLGN consisting only of AND, OR, NOT, and constant \textit{Unknown} gates.
        \end{lemma}

        \begin{proof}
            Consider the NNF tree of $\varphi$ on $[t_1,t_2]$. We construct its network by induction on tree height.
            At a leaf, a positive predicate literal is realized by its input wire, a negative literal by a NOT gate on that wire, and $U$ by a constant \textit{Unknown} gate. These realizations return the prescribed trits.
            Suppose the two child subtrees of an internal node have exact realizations, with outputs $v_1,v_2\in\{-1,0,+1\}$. Connect their roots to the polynomial gate corresponding to the node's connective $\circ\in\{\wedge,\vee\}$. By \cref{eq:polys}, its output is
            \begin{equation*}
                p_\circ(v_1,v_2)=v_1\circ v_2,
            \end{equation*}
            which is precisely the parent subtree's verdict and is again a trit. This establishes the inductive step.
            The tree is finite, so the construction reaches its root. Since normalization preserves the TTL verdict, the resulting DTLGN outputs $z^\varphi_{t_1,t_2}$ exactly.
        \end{proof}       

    \subsection{Moment Propagation Through Elementary Gates}

        Let $Z$ be a \textit{random trit} with support $\{{-}1,\,0,\,{+}1\}$. Assuming that the mean $E(Z)$ is defined, we wish to compute the propagation of the mean through the elementary ternary logic gates. Each ternary gate has a closed-form polynomial representation that is valid on $\{{-}1,\,0,\,{+}1\}$, therefore computing the mean over the connectives is equivalent to computing the mean over the polynomial surrogates. Propagating expectations through the polynomial surrogates yields $E(U)=0$ for the constant \textit{Unknown} gate and $E(\neg Z)=-E(Z)$ for negation. Conjunction for independent trits propagates through the polynomial representation~\eqref{eq:polys} according to:
        \[
            E(Z_1\wedge Z_2)=\tfrac{1}{2}E(Z_1+Z_2+Z_1Z_2-Z_1^2-Z_2^2+Z_1^2Z_2^2),
        \]
        which, with $Z_1$ and $Z_2$ independent by hypothesis, is propagated to form:
        \begin{equation}
            \label{eq:con1}
            \begin{multlined}[.85\linewidth]
                \tfrac{1}{2}\bigl(E(Z_1)+E(Z_2)+E(Z_1)E(Z_2) \\
                -E\bigl(Z_1^2\bigr)-E\bigl(Z_2^2\bigr)+E\bigl(Z_1^2\bigr)E\bigl(Z_2^2\bigr)\bigr).
            \end{multlined}
        \end{equation}
        The presence of second moments suggests that they will need to be propagated as well. For the constant \textit{Unknown} gate and negation, we find: $E\bigl(U^2\bigr)=0$ and $E\bigl({(\neg Z)}^2\bigr)=E\bigl(Z^2\bigr)$ respectively. By direct calculation with $z_1,\,z_2\in\{{-}1,\,0,\,{+}1\}$, we have:
        \[
            {(z_1\wedge z_2)}^2=\tfrac{1}{2}\bigl(z_1^2+z_2^2-z_1-z_2+z_1^2z_2+z_1z_2^2\bigr),
        \]
        and therefore the second moment for conjunction $E\bigl({(Z_1\wedge Z_2)}^2\bigr)$ is propagated according to:
        \begin{equation}
            \label{eq:con2}
            \begin{multlined}[.85\linewidth]
                \tfrac{1}{2}\bigl(E\bigl(Z_1^2\bigr)+E\bigl(Z_2^2\bigr)-E(Z_1)-E(Z_2) \\
                +E\bigl(Z_1^2\bigr)E(Z_2)+E(Z_1)E\bigl(Z_2^2\bigr)\bigr).
            \end{multlined}
        \end{equation}
        
        Similarly, calculation of the first moment for a disjunction $E(Z_1\vee Z_2)$ propagates according to:
        \begin{equation}
            \label{eq:dis1}
            \begin{multlined}[.85\linewidth]
                \tfrac{1}{2}\bigl(E(Z_1)+E(Z_2)-E(Z_1)E(Z_2) \\
                +E\bigl(Z_1^2\bigr)+E\bigl(Z_2^2\bigr)-E\bigl(Z_1^2\bigr)E\bigl(Z_2^2\bigr)\bigr),
            \end{multlined}
        \end{equation}
        and by direct calculation with $z_1,\,z_2\in\{{-}1,\,0,\,{+}1\}$, we have:
        \[
            {(z_1\vee z_2)}^2=\tfrac{1}{2}\bigl(z_1^2+z_2^2+z_1+z_2-z_1^2z_2-z_1z_2^2\bigr),
        \]
        thus allowing the second moment of disjunction $E\bigl({(Z_1\vee Z_2)}^2\bigr)$ to be propagated according to:
        \begin{equation}
            \label{eq:dis2}
            \begin{multlined}[.85\linewidth]
                \tfrac{1}{2}\bigl(E\bigl(Z_1^2\bigr)+E\bigl(Z_2^2\bigr)+E(Z_1)+E(Z_2) \\
                -E\bigl(Z_1^2\bigr)E(Z_2)-E(Z_1)E\bigl(Z_2^2\bigr)\bigr).
            \end{multlined}
        \end{equation}
        
        We now have a closed-form expression for first and second moment propagation through the elementary ternary gates. By extension, the first and second moments for every TTL formula expressed as a polynomial surrogate ternary logic network can be computed using Eqs.~\eqref{eq:con1}--\eqref{eq:dis2}.

    \subsection{Predicate Relaxation and Maximum Likelihood}

        Ternary predicate streams are, by definition, discrete. For example, consider the generic ternary predicate function~\eqref{eq:pred} with $\delta>0$. The derivative $\sfrac{\partial\mu}{\partial f(x_t)}$ is zero almost everywhere and undefined at $f(x_t)=\pm\delta$. Therefore, gradients computed through predicates vanish and provide no useful information when synthesizing controls. Generating usable gradient information is an exercise in replacing each predicate trit with a smooth representation that remains faithful to the discrete verdict. So far, no assumptions on the distribution of predicate trits were made other than the existence of the first and second moments $E(Z^\mu)$ and $E\bigl({(Z^\mu)}^2\bigr)$. To proceed, we need a relaxation of the predicate distribution that is smooth with respect to $f(x_t)$ and preserves the information carried by the input predicate streams.
    
        A similar problem arises in machine learning with discrete random variables. Samples of a categorical variable are not differentiable with respect to its parameters. The Gumbel-Softmax~\cite{jang2016categorical} and Concrete~\cite{maddison2016concrete} distributions address this by perturbing the logits of a discrete variable with noise and replacing $\argmax$ with a temperature-controlled softmax, thereby allowing gradients to be propagated through the relaxed samples. We demonstrate that ternary predicates admit the same construction.
        
        A predicate trit is obtained by evaluating $f(x_t)$ with respect to a threshold $\pm\delta$. Perturbing $f(x_t)$ with logistic noise therefore yields closed-form probabilities for each of the three verdicts. Unlike~\cite{jang2016categorical, maddison2016concrete}, we do not relax the samples themselves, rather each predicate trit remains discrete, and the distribution instead depends smoothly on $f(x_t)$. Both moments $E(Z^\mu)$ and $E\bigl({(Z^\mu)}^2\bigr)$ can be computed exactly from this distribution, thus proving compatible with the established moment propagation formulation.
    
        If we assume that each predicate function is perturbed by a small disturbance $f(x_t)+\varepsilon_t$, where $\varepsilon_t$ is drawn independently from a zero-mean logistic distribution with scale $\tau/2$, then we find the likelihoods of satisfaction and violation follow:
        \[
            \begin{aligned}
                \Pr(Z^\mu={+}1)&=\Pr\bigl(f(x_t)+\varepsilon_t\geq\delta\bigr)=\sigma\!\left(\tfrac{2(f(x_t)-\delta)}{\tau}\right), \\
                \Pr(Z^\mu={-}1)&=\Pr\bigl(f(x_t)+\varepsilon_t\leq-\delta\bigr)=\sigma\!\left(-\tfrac{2(f(x_t)+\delta)}{\tau}\right).
            \end{aligned}
        \]
        With the identity $\sigma(2x)=\frac12(1+\tanh{x})$, the first two predicate moments are:
        \[
            \begin{aligned}
                E(Z^\mu)&=\tfrac{1}{2}\Bigl(\tanh{\tfrac{f(x_t)+\delta}{\tau}}+\tanh{\tfrac{f(x_t)-\delta}{\tau}}\Bigr), \\
                E\bigl({(Z^\mu)}^2\bigr)&=1-\tfrac{1}{2}\Bigl(\tanh{\tfrac{f(x_t)+\delta}{\tau}}-\tanh{\tfrac{f(x_t)-\delta}{\tau}}\Bigr).
            \end{aligned}
        \]
        \begin{figure}[htbp]
            \centering
            \includegraphics[width=\linewidth]{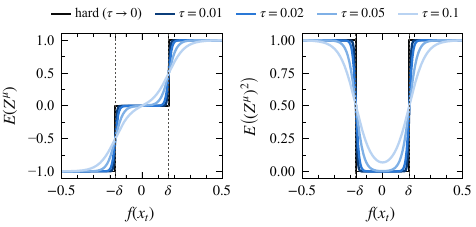}
            \caption{Example mean (left) and second moment (right) of a random predicate trit $Z^\mu$ as functions of the value $f(x_t)$ for various noise scales $\tau$.}
            \label{fig:leaf}
        \end{figure}
        
        \begin{thm}
            Given a TTL formula $\varphi$, under the predicate relaxation hypothesis, the likelihood of satisfaction of $\varphi$ can be computed using moment propagation through the $\varphi$-equivalent DTLGN. Furthermore, when the predicates are hardened to exact trits, the DTLGN computes the satisfaction of $\varphi$ exactly.
        \end{thm}
        \begin{proof}
            To construct a smooth objective that, when maximized, ensures a high likelihood of satisfaction of $\varphi$, we use the following observations:
            \[
                \begin{aligned}
                    E(Z^\varphi)&=\Pr(Z^\varphi={+}1)-\Pr(Z^\varphi={-}1), \\
                    E\bigl({(Z^\varphi)}^2\bigr)&=\Pr(Z^\varphi={+}1)+\Pr(Z^\varphi={-}1).
                \end{aligned}
            \]
            Combining the above two equations reveals that the probability of satisfaction of $\varphi$ is:
            \[
                \Pr(Z^\varphi={+}1)=\frac{E(Z^\varphi)+E\bigl({(Z^\varphi)}^2\bigr)}{2},
            \]
            where the moments $E(Z^\varphi)$ and $E\bigl({(Z^\varphi)}^2\bigr)$ are propagated through the DTLGN according to Eqs.~\eqref{eq:con1}--\eqref{eq:dis2}. We note that when the predicate trits are hardened to $\{{-}1,\,0,\,{+}1\}$, then $\frac{1}{2}\bigl(z^\varphi+{(z^\varphi)}^2\bigr)=\bm{1}(z^\varphi={+}1)$ exactly.
        \end{proof}
\section{Case Studies\label{sec:experiments}}

    Given a TTL formula $\varphi$, we have a constructive approach to build the equivalent DTLGN representing the satisfaction of $\varphi$ and contains only the elementary conjunction, disjunction, negation, and constant \textit{Unknown} gates. Moments are then propagated analytically through the DTLGN revealing a smooth closed-form expression for $\Pr(Z^\varphi={+}1)$ at any given time step over a fixed horizon. Now we demonstrate that $\log\Pr(Z^\varphi={+}1)$ provides a sound optimization objective for gradient-based control synthesis.

    \begin{table}[t]
        \centering
        \caption{Bin-picking specifications.}
        \label{tab:specs}
        \footnotesize
        \setlength{\tabcolsep}{4pt}
        \renewcommand{\arraystretch}{1.12}
        \begin{tabular}{@{}llcc@{}}
            \toprule
            \multicolumn{4}{@{}l}{\textbf{Open loop} ($N=140$)} \\
            \multicolumn{4}{@{}l}{\quad$\seq(\mathrm{Pre},\mathrm{Reach},\mathrm{Close},\mathrm{Lift},\sel(\mathrm{inspected},\mathrm{Scan}),\mathrm{Drop})\wedge\mathrm{Safe}$} \\[2pt]
            \multicolumn{4}{@{}l}{\textbf{Closed loop} ($H=50$)} \\
            \multicolumn{4}{@{}l}{\quad$\seq(\mathrm{Pre},\mathrm{Reach},\mathrm{Close},\sel(\mathrm{Place}_A,\mathrm{Place}_B))\wedge\mathrm{Safe}$} \\
            \midrule
            Stage & Condition & Open-Loop & Closed-Loop \\
            \midrule
            Pre & above part, gripper open & \checkmark & \checkmark \\
            Reach & at part, gripper open & \checkmark & \checkmark \\
            Close & at part, gripper closed & \checkmark & \checkmark \\
            Lift & above part, gripper closed & \checkmark & \\
            Scan & at camera, facing camera & \checkmark & \\
            Drop & at bin, gripper open, arm still & \checkmark & \\
            $\mathrm{Place}_A$ & at bin $A$ holding part, then open$^{*}$ & & \checkmark \\
            $\mathrm{Place}_B$ & at bin $B$ holding part, then open & & \checkmark \\
            Safe & $\square$(collision-free, closed in transit$^\dagger$) & \checkmark & \checkmark \\
            \bottomrule
        \end{tabular}
        \par\vspace{2pt}
        \parbox{\columnwidth}{\centering\scriptsize $^{*}$Only while bin $A$ is clear.\quad $^\dagger$Closed-loop must remain clear of the hand.}
    \end{table}
    
    \subsection{Open-Loop Control Synthesis}
        \begin{figure*}
            \centering
            \includegraphics[width=.9\textwidth]{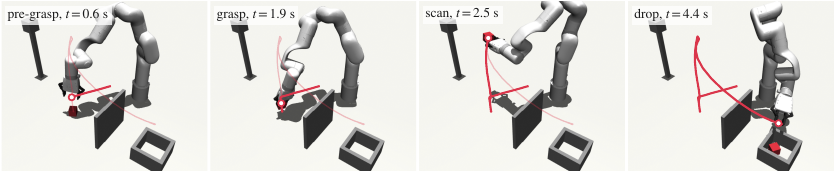}
            \caption{Open-loop plan for the uninspected cube with satisfaction likelihood $\Pr\bigl(Z^\varphi_{0,\,N}={+}1\bigr)=0.9998$. The tool path is shown in red.}
            \label{fig:ol-snapshots}
        \end{figure*}
        Consider a robotic manipulator tasked with picking items. An additional conditional constraint is added to the task: if an item has not previously been inspected, the robot must position the item in front of a scanner to be inspected before finally placing it into a bin. This task is formally represented as a TTL Sequence formula with the subtasks specified in Table~\ref{tab:specs}. The overall specification $\varphi$ is a conjunction with a safety specification stating that the robot must follow a collision-free trajectory at all times.

        We use a UFactory xArm7 $7$-DoF manipulator with its corresponding parallel gripper modeled in MuJoCo XLA (MJX)~\cite{todorov2012mujoco}. The state of the robot at any time is a vector $x_t=(q_t,\,\dot{q}_t,\,g_t)\in\mathbb{R}^{15}$ containing the joint positions, joint velocities, and gripper closure $g_t\in[0,\,1]$. The manipulator is controlled with direct joint torque commands and gripper closure commands $u_t=(\tau_t,\,u_t^g)\in\mathbb{R}^8$. We represent the satisfaction of $\varphi$ as a DTLGN and compute analytical moments through $Z^\varphi$ with noise scale annealed during solving from $\tau=0.02\to 0.003$ to formulate the log-likelihood of satisfaction as an objective. The overall optimal control problem is:
        \begin{equation}
            \label{eq:ol-ocp}
            \begin{aligned}
                \min_{\bm{u}}\;&\sum_{t=0}^{N-1}\tau_t^\top R\,\tau_t-\log\Pr\bigl(Z^{\varphi}_{0,N}={+}1\bigr)\\
                \text{s.t.}\;&x_{t+1}=f(x_t,\,u_t),\quad x_0=\xi,\\
                &|\tau_{t,i}|\leq\bar\tau_i,\quad u^g_t\in[0,1],\\
                &q_t\in[\underline{q},\,\bar q],\quad |\dot q_{t,i}|\leq 2\,\text{rad/s},
            \end{aligned}
        \end{equation}
        where $\bar{\tau}$ are the joint torque limits and $[\underline{q},\,\bar q]$ is the allowed joint limit range. Each predicate was assigned an uncertainty threshold of $\delta = 0.015$. The simulation runs for $N=140$ time steps with a time resolution of $\delta t=0.05$\,s. Problem~\eqref{eq:ol-ocp} is transcribed with multiple shooting in CasADi~\cite{andersson2019casadi} and solved with IPOPT~\cite{wachter2006implementation} using a limited-memory Hessian approximation. A satisfying plan was found after warm-starting from the best of $4$ random initializations. The reported likelihood of satisfaction reached $\Pr\bigl(Z^\varphi_{0,\,N}={+}1\bigr)=0.9998$ in $4.7$ minutes on an Ubuntu 24.04 machine with an Intel Core Ultra 9 285K CPU, an NVIDIA RTX 5090 GPU (32~GB), and 64~GB of RAM. A snapshot of the satisfying open-loop rollout is shown in Fig.~\ref{fig:ol-snapshots}.
        \begin{figure}[htbp]
            \centering
            \includegraphics[width=\linewidth]{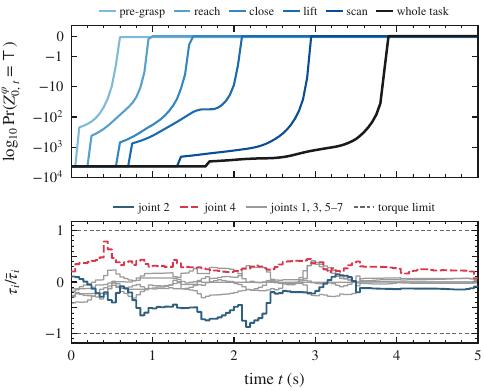}
            \caption{Log-likelihood (top) that each sequential subtask has been completed by the open-loop controller at time $t$ along the plan of Fig.~\ref{fig:ol-snapshots}. Joint torques (bottom) normalized by their maximum value. Each joint torque command remains within its respective limits.}
            \label{fig:ol-timeline}
        \end{figure}

    \subsection{Closed-Loop Control Synthesis}
        \begin{figure*}
            \centering
            \includegraphics[width=\textwidth]{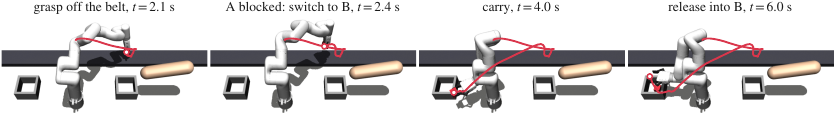}
            \caption{Rollout snapshots for the closed-loop controller. Presence of a human causes the manipulator to fall back and place the item into the secondary bin. The task specification is achieved with likelihood of satisfaction $\Pr\bigl(Z^{\varphi_\text{task}}_{0,\,N}={+}1\bigr)=0.9999$ and safety non-violation $\Pr\bigl(Z^{\varphi_\text{safe}}_{0,\,N}\neq{-}1\bigr)=0.7996$.}
            \label{fig:mpc-snap}
        \end{figure*}
        Consider the same robot manipulator from the open-loop problem. Now we task the robot with sorting items from a conveyor belt. The robot has two bins, $A$ and $B$, into which it can choose to sort items. The robot is tasked to try to place items first in bin $A$. However, a human coworker also sources items from bin $A$. The robot is constrained such that if a human is reaching into bin $A$, it should avoid contact and instead sort items into bin $B$. Because the environment is dynamic, we choose to solve the control problem online with an MPC loop. At each time step we solve the finite-horizon optimal control problem:
        \begin{equation}
            \label{eq:mpc}
            \begin{aligned}
                \min_{u_t,\dots,u_{t+H-1}}\;&\sum_{k=t}^{t+H-1}\tau_k^\top R\,\tau_k-\log\Pr\bigl(Z^{\varphi_\text{task}}_{0,\,t+H}={+}1\bigr) \\
                &\qquad\qquad\quad\;\;-\log\Pr\bigl(Z^{\varphi_\text{safe}}_{0,\,t+H}\neq{-}1\bigr)\\
                \text{s.t.}\;&x_{k+1}=f(x_k,\,u_k),\quad k=t,\dots,t+H-1,\\
                &|\tau_{k,i}|\leq\bar\tau_i,\quad u^g_k\in[0,1],\\
                &q_k\in[\underline{q},\,\bar q],\quad |\dot q_{k,i}|\leq 2\,\text{rad/s}, \\
                &x_t\text{ measured},
            \end{aligned}
        \end{equation}
        and the first element of the optimal action sequence $u_t$ is deployed and the state updated. The new state is measured and Problem~\eqref{eq:mpc} is solved again for the updated horizon. The MPC horizon is set to $H=50$ time steps with a time resolution of $\delta t=0.05$\,s. The task objective is evaluated on the measured history $(x_0,\dots,\,x_t)$ followed by the predicted states. Each MPC solve takes between $2.4$ to $3.5$\,s on average to compute on the same RTX 5090 machine, so the controller does not run in real time. Nevertheless, the overall task completed with verdict \textit{True} and likelihood of satisfaction $\Pr\bigl(Z^{\varphi_\text{task}}_{0,\,N}={+}1\bigr)=0.9999$. The safety specification verdict evaluates to \textit{Unknown} with a likelihood of non-violation $\Pr\bigl(Z^{\varphi_\text{safe}}_{0,\,N}\neq{-}1\bigr)=0.7996$ meaning the safety obligation was not violated, though it was never fully satisfied either. A snapshot of the satisfying closed-loop deployment is shown in Fig.~\ref{fig:mpc-snap}.
        \begin{figure}[htbp]
            \centering
            \includegraphics{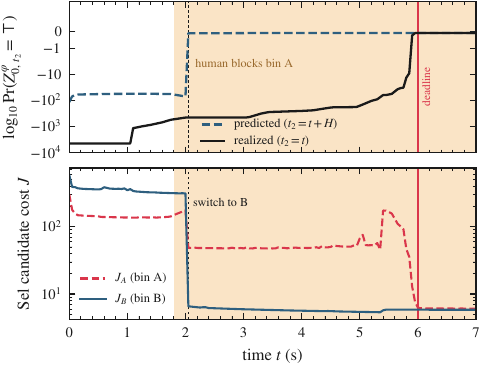}
            \caption{Log-likelihood for the actual realized rollout and the predicted inner-loop with respect to time. When a coworker reaches for bin $A$, the bin-$A$ candidate's cost rises above bin $B$'s and the manipulator switches to placing the item in bin $B$, resulting in satisfaction of the specification.}
        \end{figure}

\section{Conclusion\label{sec:conclusion}}
    This work is a continuation of our investigation into the use of ternary logic and ternary temporal logic for robotic applications. Ternary logic is a natural formalism for robotics as agents must be able to reason about their task while it is still being executed without committing to a definitive verdict. We demonstrate that any ternary temporal logic formula can be realized as a differentiable ternary logic gate network. Although innately smooth, naive predicate relaxation through a ternary logic gate network does not provide meaningful gradients for optimization. As a remedy, we propose to instead model the predicate trits as random trits whose distribution depends smoothly on the state, allowing for analytical moment propagation and a closed-form likelihood objective. The likelihood of satisfaction proves to be a useful and smooth objective in gradient-based control synthesis. Its utility is demonstrated in both an offline and an online control task for a robot manipulator, in which the task satisfaction was shown to evaluate to \textit{True} for both cases.
\bibliographystyle{mybibstyle}
\bibliography{refs}
\ifarxiv
\section{Appendix\label{sec:appendix}}
    We present the details of specifications used in both case studies (Table~\ref{tab:detailed-specs}), and the detailed experimental parameters and results for both case studies (Table~\ref{tab:results}).
    
    \begin{table*}
        \centering
        \caption{Case study specifications. Time bounds are in steps of $\delta t=0.05$\,s, and all predicates use $\delta=0.01$. The optimizer uses $\delta_{\mathrm{opt}}=0.015$ and, in the closed-loop experiment, every $\square$ window lengthened by 3 steps. Verdicts are certified on the specifications as listed.}
        \label{tab:detailed-specs}
        \footnotesize
        \renewcommand{\arraystretch}{1.15}
        \begin{tabular*}{\textwidth}{@{}l@{\hspace{6pt}\extracolsep{\fill}}ll@{}}
            \toprule
            & \textbf{Open-loop}: pick, inspect, and sort ($N=100$) & \textbf{Closed-loop}: conveyor pick with a moving hand ($H=50$, $N=160$) \\
            \midrule
            Pre & $\lozenge_{[0,N]}\square_{[0,2]}(\mathrm{at}(c_{\mathrm{pre}},0.03)\wedge\mathrm{open}\wedge\mathrm{aligned})$
                & $\lozenge_{[0,120]}\square_{[0,2]}(\mathrm{at}(c_{\mathrm{part}}{+}0.1e_z,0.03)\wedge\mathrm{open}\wedge\mathrm{aligned})$ \\
            Reach & $\lozenge_{[0,N]}\square_{[0,2]}(\mathrm{at}(c_{\mathrm{grasp}},0.025)\wedge\mathrm{open}\wedge\mathrm{aligned})$
                & $\lozenge_{[0,30]}\square_{[0,2]}(\mathrm{at}(c_{\mathrm{grasp}},0.025)\wedge\mathrm{open}\wedge\mathrm{aligned}\wedge\mathrm{belt})$ \\
            Close & $\lozenge_{[0,N]}\square_{[0,6]}(\mathrm{at}(c_{\mathrm{grasp}},0.025)\wedge\mathrm{closed}\wedge\mathrm{aligned})$
                & $\lozenge_{[0,20]}\square_{[0,6]}(\mathrm{at}(c_{\mathrm{grasp}},0.025)\wedge\mathrm{closed}\wedge\mathrm{aligned}\wedge\mathrm{belt})$ \\
            Lift & $\lozenge_{[0,N]}\square_{[0,2]}(\mathrm{at}(c_{\mathrm{pre}},0.03)\wedge\mathrm{closed})$ & -- \\
            Inspect & $\sel\bigl(\mathrm{inspected},\,\lozenge_{[0,N]}\square_{[0,10]}(\mathrm{at}(c_{\mathrm{scan}},0.04)\wedge\mathrm{facing})\bigr)$ & -- \\
            Place & $\lozenge_{[0,N]}\square_{[0,6]}(\mathrm{at}(c_{\mathrm{drop}},0.04)\wedge\mathrm{open}\wedge\mathrm{still})$
                & $\sel(P_A,P_B)$, \;$P_A=\seq\bigl(\lozenge_{[0,60]}\square_{[0,1]}(\mathrm{at}(c_A,0.03)\wedge\mathrm{hold}\wedge\mathrm{free}_A),$ \\
                & & \hspace{1.2em}$\lozenge_{[0,20]}\square_{[0,6]}(\mathrm{at}(c_A,0.03)\wedge\mathrm{open}\wedge\mathrm{still})\bigr)$; \;$P_B$: same with $c_B$, without $\mathrm{free}_A$ \\
            Task & $\seq$ of the stages above
                & $\seq$ of the stages above $\wedge\;\lozenge_{[0,120]}\bigl((\mathrm{at}(c_A,0.03)\wedge\mathrm{open}\wedge\mathrm{still}\wedge\mathrm{free}_A)$ \\
                & & \hspace{1.2em}$\vee\,(\mathrm{at}(c_B,0.03)\wedge\mathrm{open}\wedge\mathrm{still})\bigr)$ \\
            Safety & $\square_{[0,N]}\bigl(\mathrm{clear}\wedge\mathrm{tool}\wedge(\mathrm{high}\vee\mathrm{column})$
                & $\square_{[0,N]}\chi$, \;$\chi=\mathrm{clear}\wedge\mathrm{clear}_H\wedge\mathrm{tool}\wedge(\mathrm{high}\vee\mathrm{column})$ \\
                & \quad$\wedge\,(\neg\mathrm{open}\vee\mathrm{at}(c_{\mathrm{pre}},0.15)\vee(\mathrm{at}(c_{\mathrm{drop}},0.08)\wedge\mathrm{still}))\bigr)$
                & \quad$\wedge\,(\neg\mathrm{open}\vee\mathrm{belt}\vee(\mathrm{at}(c_A,0.08)\wedge\mathrm{still}\wedge\mathrm{free}_A)\vee(\mathrm{at}(c_B,0.08)\wedge\mathrm{still}))$ \\
            Objective & $\max\,\log\Pr(Z^{\mathrm{Task}\wedge\mathrm{Safety}}_{0,N}={+}1)$
                & $\max\,\log\Pr(Z^{\square_{[0,t+H]}\chi}_{0,t+H}={+}1)+\log\Pr(Z^{\mathrm{Task}}_{0,t+H}={+}1)$ at each step $t$ \\
            \midrule
        \end{tabular*}
    
        \renewcommand{\arraystretch}{1.1}
        \begin{tabular*}{\textwidth}{@{}l@{\extracolsep{\fill}}lllll@{}}
            \multicolumn{6}{@{}l}{\textbf{Predicates} (\textit{True} iff $f\geq\delta$)} \\[1pt]
            $\mathrm{at}(c,r)$ & $r-\|p-c\|$ &
            $\mathrm{open}$ & $0.1(0.25-g)$ &
            $\mathrm{closed}$ & $0.1(g-0.75)$ \\
            $\mathrm{down}$ & $\kappa(12^\circ-\angle(Re_3,-e_z))$ &
            $\mathrm{yaw}$ & $\kappa(15^\circ-|\psi-\psi^\star|_{90^\circ})$ &
            $\mathrm{aligned}$ & $\mathrm{down}\wedge\mathrm{yaw}$ \\
            $\mathrm{facing}$ & $\kappa(20^\circ-\angle(Re_3,v_{\mathrm{cam}}))$ &
            $\mathrm{still}$ & $0.1(0.4-\|\dot q\|)$ &
            $\mathrm{high}$ & $p_z-z_{\min}$ \\
            $\mathrm{clear}$ & $\min_{i,\mathcal{B}}(d_{\mathcal{B}}(s_i)-\rho_i)-0.02$ &
            $\mathrm{tool}$ & $\min_{\mathcal{B}}d_{\mathcal{B}}(p)-0.04$ &
            $\mathrm{column}$ & $0.03-\|p_{xy}-c_{xy}\|$ \\
            $\mathrm{clear}_H$ & $\min(\min_i d(s_i,\mathcal{H})-\rho_i,\,d(p,\mathcal{H})-0.04)-r_{\mathcal{H}}-0.05$ &
            $\mathrm{free}_A$ & $d(c_A,\mathcal{H})-r_{\mathcal{H}}-0.15$ &
            $\mathrm{hold}$ & $0.03-\|c_{\mathrm{grasp}}-p\|$ \\
            $\mathrm{belt}$ & $z_{\mathrm{belt}}+0.015-c_{\mathrm{part},z}$ &
            $\mathrm{inspected}$ & constant & & \\
            \bottomrule
        \end{tabular*}
    
        \vspace{2pt}
        \parbox{\textwidth}{\scriptsize $p$, $R$: tool position and orientation; $g\in[0,1]$: gripper closure; $\dot q$: joint velocities; $s_i$, $\rho_i$: link collision spheres; $d_{\mathcal{B}}$: signed distance to obstacle $\mathcal{B}$ (including the floor; $\mathrm{tool}$ uses only the bins in the closed-loop); $\mathcal{H}$, $r_{\mathcal{H}}$: forearm segment and radius of the human hand; $c_{\mathrm{part}}$: center of the moving part; $c_{\mathrm{grasp}}$: grasp point ($c_{\mathrm{part}}-0.013e_z$ in the closed-loop); $\psi$, $\psi^\star$: tool and object yaw, $|\cdot|_{90^\circ}$ wrapped to $(-45^\circ,45^\circ]$; $\kappa=0.2$\,m/rad. Distances in meters.}
    \end{table*}

    \begin{table*}[t]
    \centering
    \caption{Detailed simulation parameters and results. Verdicts are exact $K_3$ evaluations, probabilities
    are the smooth semantics with nominal $\delta=0.01$.
    Both experiments ran on NVIDIA RTX~5090 GPU, Intel Core Ultra~9 285K CPU.}
    \label{tab:results}
    \footnotesize
    \renewcommand{\arraystretch}{1.15}
    \begin{tabular*}{\textwidth}{@{}l@{\hspace{6pt}\extracolsep{\fill}}ll@{}}
        \toprule
        & \textbf{Open-loop}: pick, scan, and sort & \textbf{Closed-loop}: conveyor pick, receding horizon \\
        \midrule
        \multicolumn{3}{@{}l}{\textit{Problem}} \\
        Robot & \multicolumn{2}{@{}l}{UFactory xArm7 with gripper (MJX), torque control} \\
        State, input & \multicolumn{2}{@{}l}{$x=(q,\dot q,g)\in\R^{15}$, $|\dot q_i|\le2$\,rad/s;\;
            $u=(\tau,u_g)\in\R^{8}$, $|\tau_i|\le\bar\tau_i$, $\bar\tau=(50,50,30,30,30,20,20)$\,Nm} \\
        Dynamics & \multicolumn{2}{@{}l}{MJX, contact-free, smoothed Coulomb friction; $\delta t=50$\,ms} \\
        Specification & $\seq$ of 6 stages with $\sel(\mathrm{inspected},\mathrm{Scan})$
            & $\seq$ of 4 stages with $\sel(\mathrm{Place}_A,\mathrm{Place}_B)$, deadline 6\,s \\
        AST nodes / predicates & 60 / 17 & 82 / 23 (Task 58 / 12, Safe 24 / 11) \\
        Horizon & $N=100$ (5\,s) & $H=50$ (2.5\,s) per solve; 160 steps (8\,s) in closed-loop \\
        Environment & static part, scan station, two bins
            & belt at 0.08~m/s, forearm blocks bin $A$ for $t\in[2,7]$\,s \\
        \midrule
        \multicolumn{3}{@{}l}{\textit{Optimization}} \\
        Transcription & \multicolumn{2}{@{}l}{CasADi multiple shooting, IPOPT (L-BFGS), all derivatives by JAX and MJX} \\
        Objective & $\sum_t\tau_t^\top R\,\tau_t-\log\Pr(Z^{\varphi}_{0,N}={+}1)$
            & $\sum_t\tau_t^\top R\,\tau_t-\log\Pr(Z^{\square\mathrm{Safe}}_{0,t+H}={+}1)-\log\Pr(Z^{\mathrm{Task}}_{0,t+H}={+}1)$ \\
        Variables / equalities & 2315 / 1515 & 1165 / 765 per $\sel$ candidate (2 candidates per step) \\
        Relaxation & $\tau=0.005\to0.003$,\; $\delta_{\mathrm{opt}}=0.015$
            & $\tau=0.01$,\; $\delta_{\mathrm{opt}}=0.015$,\; hold windows padded by 3 steps \\
        IPOPT iterations & 1600 ($2\times800$, $\tau$ annealing) & $\le15$ per solve; 3 solves per candidate before the grasp \\
        Initialization & previous solution & shifted previous solution \\
        \midrule
        \multicolumn{3}{@{}l}{\textit{Computation}} \\
        Solve time & 4.7\,min (0.17\,s per iteration) & 2.52\,s mean per step \\
        \midrule
        \multicolumn{3}{@{}l}{\textit{Results}} \\
        Verdict (Task / Safe) & \textit{True}\ / \textit{True}\ on the MJX rollout of the plan
            & \textit{True}\ / \textit{Unknown}\ on the executed closed-loop \\
        Smooth probability & $\Pr(Z^{\varphi}_{0,N}={+}1)=0.9998$ ($\tau=0.003$)
            & Task $0.9999$;\; Safe $\Pr(Z\neq{-}1)=0.7996$ ($\tau=0.01$) \\
        Events & grasp 1.9\,s, release 4.4\,s & grasp 2.1\,s, switch $A\to B$ 2.4\,s, release 6.0\,s \\
        Min.\ clearance margin & obstacles 28.2\,mm & obstacles 34.7\,mm, human 155\,mm \\
        Peak $|\tau_i|/\bar\tau_i$, $|\dot q_i|$ & 0.87, 2.0\,rad/s & 0.99, 2.0\,rad/s \\
        \bottomrule
    \end{tabular*}
\end{table*}

\fi
\end{document}